\documentclass[letterpaper, 10pt, conference]{ieeeconf}

\IEEEoverridecommandlockouts
\usepackage{graphicx}
\usepackage{psfrag}
\graphicspath{{./Figures/}}
\usepackage{amsmath,amssymb,comment}
\usepackage{multirow}
\usepackage{todonotes}
\usepackage{url}
\usepackage{color}

\newcommand{\eg}{{\it e.g., }}

\newcommand{\ie}{{\it i.e., }}

\newcommand{\reals}{{\mathbb{R}}}

\newcommand{\rank}{{\mathrm{rank}\,}}
\newcommand{\normrank}{{\mathrm{normalrank}\,}}

\newtheorem{defn}{Definition}
\newtheorem{thm}{Theorem}
\newtheorem{rem}{Remark}
\newtheorem{prop}{Proposition}
\newtheorem{lem}{Lemma}
\newtheorem{cor}{Corollary}

\title{\LARGE \bf
From Control System Security Indices to Attack Identifiability
}

\author{Henrik Sandberg and Andr\'e~M.H.~Teixeira
\thanks{This work was supported in part by the Swedish Research Council (grant~2013-5523), the Swedish Civil Contingencies Agency through the CERCES project, and the EU 7th Framework Programme (FP7/2007-2013, grant agreement n$^\circ$ 608224).}
\thanks{H.~Sandberg is with the Department of Automatic Control, School of Electrical Engineering, KTH Royal Institute of Technology, Stockholm, Sweden. Email: {\tt hsan@kth.se}}
\thanks{A.M.H.~Teixeira is with the Faculty of Technology, Policy and Management, Delft University of Technology, Delft, the Netherlands. Email: {\tt andre.teixeira@tudelft.nl}}
}

\begin{document}
\maketitle
\thispagestyle{empty}
\pagestyle{empty}

\begin{abstract}
In this paper, we investigate detectability and identifiability of attacks on linear dynamical systems that are subjected
to external disturbances. We generalize a concept for a security index, which was previously introduced for
static systems. The generalized index exactly quantifies the resources necessary for targeted attacks to be undetectable and unidentifiable in the presence of disturbances. This information is useful for both risk assessment and for the design of anomaly detectors.
Finally, we show how techniques from the fault detection literature can be used to decouple disturbances and to identify attacks, under certain sparsity constraints.
\end{abstract}

\section{Introduction}
\label{sec:intro}
As modern control systems increasingly rely on information and communication technology (ICT) infrastructures to exchange real-time measurements and actuator signals, their exposure to malicious cyber threats also grows: each measurement and actuator signal may be compromised and altered by a skillful cyber adversary.
Therefore, cyber security and resilience with respect to attacks are important properties of modern control systems that are tightly coupled to ICT infrastructures.

Some of the main challenges in designing cyber-secure control systems are related to: analyzing the risk of cyber attacks; devising protection mechanisms to prevent and remove high-risk threats; and also to timely detect and mitigate on-going attacks. While the first two challenges relate to conventional ICT cyber security approaches (\ie risk management~\cite{teixeira+15}), the third approach is closely related to the well-known control field of fault diagnosis. Although both relate to detecting anomalies, there exist subtle differences between classical fault diagnosis and attack detection in cyber security.
Classical control-theoretic approaches to anomaly detection (\eg fault detection, isolation, and identification) typically deal with independent disturbances and faults; thus they typically do not consider possibly colluding malicious cyber attacks, which may even attempt to hide the attacks by mimicking physical disturbances and faults. In fact, this paper addresses the latter scenario, discussing detectability conditions of sparse attacks that may be masked by plausible disturbances, and connecting the results to fundamental limitations well-known in the controls literature, in terms of fault detection and identification~\cite{ding08} and input reconstruction~\cite{patton+99}.

The topic of cyber-secure control systems has been receiving increasing attention recently. An overview of existing cyber threats and vulnerabilities in networked control systems is presented in \cite{cardenas+08,Teixeira_Automatica2015}. Rational adversary models are highlighted as one of the key items in security for control systems, thus making adversaries endowed with intelligence and intent, as opposed to faults. Therefore, these adversaries may exploit existing vulnerabilities and limitations in the traditional anomaly detection mechanisms and remain undetected, or indistinguishable from disturbances and process noise. In fact,~\cite{pasqualetti+13} uses such fundamental limitations to characterize a set of undetectable attack policies for networked systems modeled by differential-algebraic equations. Related undetectable attack policies were also considered in \cite{Teixeira_Automatica2015,smith+15}. A common thread within these approaches is that undetectable attacks are constrained to be entirely decoupled from the anomaly detector's output.

Detectability conditions of undetectable false-data injection attacks to control systems are closely examined in~\cite{Teixeira_Allerton2012}, where it is shown that mismatches between the system's and the attack policy's initial conditions may lead to detectable attacks. Additionally, modifications to the system dynamics, input, and output matrices that reveal stealthy data attacks were also characterized.

Other work has analyzed undetectable attacks with respect to the amount of effort they require, \ie the number of attack signals that must be injected by the adversary to remain undetected. As discussed in~\cite{teixeira+15}, such analysis provides insight into the likelihood of such attacks occurring, which is a core component of determining the risk (\ie impact and likelihood) of such threat scenarios.

For static systems,~\cite{sandberg+10} first proposed a security index for measurement attacks, which corresponds to the minimum number of measurements that need to be corrupted as to ensure undetectability. The computation of the security index involves solving an NP-hard problem, in general, which has later been investigated by~\cite{sou+12,hendrickx+14,kosut14,yamaguchi+15}.
Under certain structures of the problem, this work proposed efficient algorithms to compute the security index in polynomial time.

Related problems have been investigated for dynamical systems. The work in~\cite{fawzi+14} characterizes the number of corrupted sensor channels that cannot be detected during a finite time-interval. For sensor attacks that can be detected, a resilient state estimation scheme inspired by compressed sensing is proposed. The work in~\cite{sundaram+11} explored the notion of strong observability to characterize the conditions for which the initial state can be recovered under the presence of sparse unknown input signals. For sensor attack scenarios, \cite{chen+15} determines the
smallest number of sensors needed for undetectable attacks. The notion of security index for dynamical system under sensor and actuator attacks was also extended to dynamical systems at steady-state and for finite-time intervals in~\cite{teixeira+13}.

This work investigates the notion of security index for dynamical systems under both attacks and disturbances. In particular, we consider the case where attacks are said to be undetectable if they can be masked (explained) by a disturbance signal. The formulation of the security index is related to well-known limitations in the fault detection literature, and the complexity of computing these indices for special cases is discussed and related to the literature. For detectable attacks, the concept of identifiable attacks is defined, as well as a weaker notion of identifiability where only certain entries of the attack signal can be uniquely determined. Connections between these definitions and the security index are investigated, based on which attacks with sufficiently high sparsity are shown to be identifiable. Finally, for identifiable attacks, an attack reconstruction procedure is proposed.

The outline of this paper is as follows. The dynamical system under the influence of disturbances and attacks is described in Section~\ref{sec:prel}, where undetectable attacks, potentially masked by disturbances, are defined. Section~\ref{sec:index} formulates the security index for dynamical systems under the influence of both disturbances and attacks, and discusses important special cases and their connection to the literature.
The role of security indices in (possibly partial) attack identification under disturbances is examined in Section~\ref{sec:identification}, whereas concluding remarks are given in Section~\ref{sec:conclusion}.

\subsubsection*{Notation}
For a set $I$, $|I|$ denotes its cardinality. For a vector $a\in\mathbb{C}^m$, we denote
its $i$-th element by $a_i$. By $a^i \in \mathbb{C}^m$, we mean a vector whose $i$-th element is non-zero, \ie $a_i \neq 0$, and the other elements are arbitrary.
The support of $a\in\mathbb{C}^m$, $\text{supp}(a)$, is the set of indices $i$ where $a_i \neq 0$, and $\|a\|_0:=|\text{supp}(a)| $ is the number of non-zero elements in $a$. Similar notations are used for discrete-time signals $a$, where $a(k)\in\mathbb{C}^m$, $k=0,1,2,\ldots$ A discrete-time linear system $G\in\mathcal{R}_p^{p\times m}(z)$ has a rational proper transfer matrix $G(z)$ of dimension $p \times m$. We also define $\normrank [G(z)] := \max_z \rank [G(z)]$.

\section{Preliminaries}
\label{sec:prel}
Let us consider the discrete-time system $y=G_d d + G_a a$, $G_d\in\mathcal{R}_p^{p\times o}(z)$ and $G_a\in\mathcal{R}_p^{p\times m}(z)$, with a realization
\begin{equation}
\begin{aligned}
x(k+1) & = Ax(k) + B_d d(k) + B_a a(k) \\
y(k) & = Cx(k) + D_d d(k) + D_a a(k),
\end{aligned}
\label{eq:G}
\end{equation}
for times $k =0,1,2\ldots$
Here $x(k)\in \reals^n$ is the state vector, $d(k)\in\reals^o$ are unknown disturbance (or fault) signals, $a(k)\in\reals^{m}$ are potential attack signals, and $y(k)\in\reals^{p}$ are the measurements available to the operator of the system.
Additionally, we assume to have distinct measurement, attack, and disturbances signals, in the sense that
\begin{equation*}
\begin{aligned}
\rank\begin{bmatrix}
B_d  \\
D_d
\end{bmatrix} &= o,\quad
\rank\begin{bmatrix}
B_a \\
 D_a
\end{bmatrix} = m, \quad
\rank \left[C\right] = p.
\end{aligned}
\end{equation*}
It turns out that the value of the initial state, $x(0)$, is important in the following, but initially we will let it be a free variable.

The system model \eqref{eq:G} is similar to those studied in the fault detection and diagnosis literature, see, \eg \cite{patton+99,ding08}. The signals $d$ and $a$ represent different types of anomalies that can occur in the system, although of different nature. We next want to determine when we can detect and distinguish between these anomalies.
We could think of $d$ as natural disturbances, or faults, that are to be expected, and that have no malicious intent. They could represent measurement and process noise, for example. One important aspect is that a malicious attacker could use such disturbances to hide his or her attack $a$ from being seen in the output $y$. We will typically let $d$ be a free variable, where the only available knowledge about the disturbance signals amounts to their signature matrices, $B_d$ and $D_d$. Thus, to ensure robustness with respect to disturbances, an anomaly detection algorithm wishing to detect potential attack signals must be designed so that it is decoupled from $B_d$ and $D_d$.
Under this disturbance model, we check whether a disturbance exists that will "mask" the attack. If this is the case, the operator is not able to distinguish between attacks and disturbances, and cannot conclude whether an attack is present, or not.

\begin{rem}
Naturally, several other disturbance models exist, such as assuming known upper bounds on the disturbance signal's energy or instantaneous peak, or constraining the disturbance to belong to a given class of signals, \eg constant or sinusoidal signals. In particular, the results in this paper can be straightforwardly extended to disturbances modeled as the output of an autonomous discrete-time system
\begin{equation*}
\begin{aligned}
x_d(k+1) & = A_d x_d(k) \\
d(k) & = C_d x_d(k),
\end{aligned}
\label{eq:Disturbance}
\end{equation*}
which is parametrized by a free initial condition $x_d(0)$.
\end{rem}

The attack can potentially occur in $m$ different locations in the system ($a(k)\in\mathbb{R}^m$), and we will be concerned about the possibility for the operator with access to the above model and the signal $y$ to detect an attack signal $a\neq 0$. We make the following definitions to formalize these ideas.

\begin{defn}
An attack signal $a$ is \emph{persistent} when $a(k) \not \rightarrow 0$ as $k\rightarrow \infty$.
\end{defn}

In this paper, we are mainly concerned with persistent attacks, since they have non-vanishing impact.

\begin{defn}
A (persistent) attack signal $a$ is
\begin{itemize}
\item[(i)] \emph{undetectable} if there exists a simultaneous (masking) disturbance signal $d$ and initial state $x(0)$ such that $y(k)=0$, $k\geq 0$;
\item[(ii)] \emph{asymptotically undetectable} if there exists a simultaneous (masking) disturbance signal $d$ and initial state $x(0)$ such that $y(k)\rightarrow 0$, $k\rightarrow \infty$.
\end{itemize}
\end{defn}

Note that the definition of undetectable attacks is the same as in~\cite{pasqualetti+13}, if we assume there are no disturbances in the system~\eqref{eq:G}. The reason for calling the disturbance "masking" comes from linearity of the system: If $0=G_a a + G_d d$, then clearly $y=G_a a = -G_d d$, and it is impossible to in the output distinguish between the undetectable attack and the masking disturbance, if they occur by themselves without the other.

We will next be interested in quantifying the minimal resources needed by the attacker to achieve undetectability, when he or she want to target a specific attack element $a_i$, $i\in\{1,\ldots,m\}$. Hence, we will search for sparse signals $a^i$ satisfying the above conditions.

\section{The Dynamical Security Index}
\label{sec:index}

For an attack signal $a$ to be undetectable, we need to ensure there exists a masking disturbance $d$ and an initial state $x(0)$ resulting in zero output. Existence of such a signal can easily be checked by considering the matrix pencil (the Rosenbrock system matrix)
\begin{equation*}
P(z) = \begin{bmatrix} A-zI & B_d & B_a \\ C & D_d & D_a \end{bmatrix},
\end{equation*}
see~\cite{ZDG96}. An attack signal $a(k)=z_0^k a_0$, $a_0\in\mathbb{C}^m$, $z_0\in\mathbb{C}$, is undetectable iff there exists
$x_0\in\mathbb{C}^n$ and $d_0\in\mathbb{C}^o$ such that
\begin{equation}
P(z_0)\begin{bmatrix} x_0 \\ d_0 \\ a_0\end{bmatrix} = 0, \label{eq:z0}
\end{equation}
\ie $P(z_0)$ should not have full column rank. The undetectable attack is also persistent iff $|z_0|\geq 1$.

\begin{rem}
Note that, if the initial state $x(0)\neq x_0$, the attack signal $a(k)=z_0^k a_0$ may actually be detectable. Following the analysis in~\cite{Teixeira_Allerton2012}, if $A$ is Schur ($\rho(A)<1$), the attack signal is only \emph{asymptotically} undetectable, since there will be a vanishing transient visible in the output. This transient can be made arbitrarily small by the attacker choosing $a_0$ small. Hence, the difference between asymptotically undetectable and undetectable attacks may not be very large in practice.
\end{rem}

If the attacker would like to target the element $i$, \ie $a_i\neq 0$, and remain undetected, he or she needs to
find a vector $a_0^i\in\mathbb{C}^m$
satisfying \eqref{eq:z0}. In general, this may require the attacker to target several elements
$a_j$, $j\neq i$. To measure the minimal number of elements required to achieve this, we introduce the following \emph{security index} $\alpha_i$, which generalizes a concept first introduced for non-dynamical systems in \cite{sandberg+10}:
\begin{equation}
\begin{aligned}
\alpha_i := & \min_{|z_0|\geq 1,x_0,d_0,a_0^i} &&  \|a_0^i\|_0 \\
& \mathrm{subject \, to} & &  P(z_0)\begin{bmatrix} x_0 \\ d_0 \\ a_0^i\end{bmatrix} = 0.
\end{aligned}
\label{eq:alpha}
\end{equation}
Note that for all $i$ it holds $\alpha_i \geq 1$, and if there is no feasible solution, we define $\alpha_i = +\infty.$
Note also that this is a combinatorial optimization problem, because of the objective function $\|a_0^i\|_0$,
and in general is hard to solve~\cite{hendrickx+14}. However, in several cases of interest, it has a simple solution, as discussed below.

We can think of the signals $a^i(k)= z_0^k a_0^i$ resulting from \eqref{eq:alpha}
as the sparsest possible persistent undetectable attacks against an element $i$.
These signals should be of interest to both the operator and the attacker, in the sense that they show how the attacker can modify the solutions to the system equations~\eqref{eq:G}, without modifying the measurable output $y$. Also, if the number $\alpha_i$ is large, it indicates that it will require significant coordinated resources by the attacker to accomplish undetectable attacks against $a_i$.
An operator can thus use the
index in performing a quantitative risk assessment, as illustrated in, \eg \cite{teixeira+15}. The index $\alpha_i$ also has implications for the possibility of the operator to reconstruct ("identify") a detectable attack $a^i$, as will be further explored in Section~\ref{sec:identification}.

\begin{rem}
There are some concepts in the literature that are similar to $\alpha_i$ above.
In power system observability analysis, a related concept is that of critical $k$-tuples, see, \eg \cite{sou+12}. For sensor attack scenarios, \cite{chen+15} determines the
smallest number of sensors needed for undetectable attacks. There are also close connections to the spark of a matrix, used in compressed sensing, see, \eg \cite{donoho+03}. Also, in \cite{teixeira+13}, an optimization problem related to \eqref{eq:alpha} was studied.
Some further connections are made in the special cases considered in the following subsections.
\end{rem}

\subsection{Critical Attack Signals ($\alpha_i = 1$)}
A particularly serious situation is when $\alpha_i=1$, since the attacker then can target element $i$ undetected without the need to access any other resources. Let us denote
\begin{equation*}
\begin{aligned}
P_i(z) & = \begin{bmatrix} A-zI & B_d & B_{a,i} \\ C & D_d & D_{a,i} \end{bmatrix} \in\mathbb{C}^{(n+p)\times (n+o+1)}, \\
P_d(z) & = \begin{bmatrix} A-zI & B_d  \\ C & D_d \end{bmatrix} \in\mathbb{C}^{(n+p)\times (n+o)},
\end{aligned}
\end{equation*}
where $B_{a,i},D_{a,i}$ are the $i$-th columns of $B_a,D_a$.
If there is a $z_0\in\mathbb{C}$, $|z_0|\geq 1$, such that
\begin{equation*}
\rank [P_d(z_0)] = \rank [P_i(z_0)],
\end{equation*}
then $\alpha_i= 1$. An even more serious situation occurs when
\begin{equation}
\normrank [P_d(z)] = \normrank [P_i(z)].
\label{eq:seriousrank}
\end{equation}
If this easily checked condition is fulfilled, it is possible to find an undetectable attack signal $a^i(k)=z_0^k a_0^i $ of cardinality one, using any complex frequency $z_0$.

Note that \eqref{eq:seriousrank} holds when there are many disturbances in relation to the number of available measurements,
\ie $o\geq p$.

\subsection{Transmission Zeros}
If the Rosenbrock system matrix $P(z)$ has full column normal rank and the realization is minimal, the only solutions to \eqref{eq:z0} that exist correspond to the system's finite number of transmission zeros, see, \eg \cite{ZDG96}. Hence, to find $\alpha_i$ we only need to inspect the corresponding system zero directions. When the zero directions are all one-dimensional, the computation of $\alpha_i$ becomes especially simple. Generically, one would expect the zero directions to be one-dimension, but there are several interesting situations where this is not the case, as we shall see below (although these will be invariant zeros, and not transmission zeros).

\subsection{Sensor Attacks}
\label{sec:sensattack}
The situation where the system is subjected to sensor attacks have received particular interest in the literature, see, \eg \cite{fawzi+14,chen+15,lee+15}. In this case we have $B_d=B_a=0$, and in \eqref{eq:z0} we only need to consider $z_0\in\{\lambda_1(A),\ldots,\lambda_n(A)\}$, \ie the eigenvalues of $A$, where $x_0$ are eigenvectors of $A$. If the eigenvalues are simple, the eigenspace corresponding to each eigenvalue is one-dimensional, and again the computation of $\alpha_i$ is simplified.

As a further special case, suppose all sensors are potentially attackable and there are no disturbances, and so $D_a=I_p$ and $D_d=0$. Also suppose that
the operator has high redundancy in the system  in the sense that
the realization \eqref{eq:G} is observable using any one of the outputs $y_i$, $i\in\{1,\ldots,p\}$, by itself.
Considering the PBH test~\cite{ZDG96}, this means that any one of the eigenmodes $z_0^kx_0$ is visible in all the sensors, and all elements in the vector $C x_0$ are non-zero. Thus an undetectable attack ($Cx_0+a_0^i=0)$ must involve all the signals in $a$, and for all $i$ the security index must be $\alpha_i = m = p$ (or $\alpha_i = +\infty$ if $A$ is Schur).
Hence, one way to make undetectable attacks hard is to
install many redundant sensors, each of which with the individual power to observe the entire system state with little noise, which is in agreement with \cite{fawzi+14}.

\subsection{Sensor Attacks for Static Systems}
If we assume $A=I$, $B_d=B_a=D_d=0$ (only sensors attacked), we have essentially recovered the original security index $\alpha_i$, as defined in \cite{sandberg+10}. The motivation for the index there was to quantify the vulnerability of power system state estimators to false data injection attacks. Note that because $A=I$ and $B_d=B_a=0$, this problem only concerns systems in steady-state. Perhaps one would think that this makes the problem \eqref{eq:alpha} easier, but in fact the problem can be significantly harder in practice. This is because the dimension of the eigenspace corresponding to the sole eigenvalue is of dimension $n$, and not one-dimensional as is frequently the case in the previous examples.
Intuitively, one can understand this since the attacker here has no constraints in time to fulfill and thus has many more options for being undetectable. This fact together with the potentially high dimension $n$ in a power system has spurred several investigations on the efficient computation of $\alpha_i$. The problem in general is NP-hard~\cite{hendrickx+14}, but in the examples deriving from power systems the matrix $C$ has a useful structure that can be exploited. In particular, \cite{hendrickx+14,kosut14,yamaguchi+15} show how max-flow min-cut algorithms can be used to solve the problem in polynomial time.
Under slightly different assumptions on the structure of $C$, \cite{sou+13} shows how $\ell_1$-relaxation can also exactly solve the problem in polynomial time.

\section{Attack Identification and Decoupling}
\label{sec:identification}
In this section, we turn to the related problem of attack identification, which concerns the
possibilities to reconstruct elements of an attack signal from the measured output.

\subsection{Attack Identification}
To formalize the attack identification problem, the following definitions are made.
\begin{defn}
A (persistent) attack signal $a$ is
\begin{itemize}
\item[(i)] \emph{identifiable} if for all attack signals $\tilde a\neq a$, and all corresponding disturbances $d$ and $\tilde d$ and initial states $x(0)$ and $\tilde x(0)$, we have $\tilde y \neq y$;
\item[(ii)] \emph{asymptotically} \emph{identifiable} if for all attack signals $\tilde a(k)\not\rightarrow a(k)$, and all corresponding disturbances $d$ and $\tilde d$ and initial states $x(0)$ and $\tilde x(0)$, we have  $\tilde y(k)\not \rightarrow y(k)$, as $k\rightarrow \infty$.
\end{itemize}
\end{defn}

Identifiable attack signals $a$ excite the output $y$ in a unique way that cannot be achieved by any other attack $\tilde a$.
This is equivalent to the system possessing a certain left inverse, as will be explored in Section~\ref{sec:decouple}.
Note that identifiability of $a$ is a much stronger requirement than detectability of $a$ (which means that the attack $a$ is such that $y\neq 0$ for all disturbances $d$ and initial states $x(0)$). Since identifiability is such a strong requirement, we will also be interested in the following weaker notion.

\begin{defn}
A (persistent) attack signal $a$ is
\begin{itemize}
\item[(i)] $i$-\emph{identifiable} if for all attack signals $\tilde a$ with $\tilde a_i\neq a_i$, and all corresponding disturbances $d$ and $\tilde d$ and initial states $x(0)$ and $\tilde x(0)$,
we have $\tilde y \neq y$;
\item[(ii)] \emph{asymptotically} $i$-\emph{identifiable} if for all attack signals $\tilde a$ with $\tilde a_i(k)\not\rightarrow a_i(k)$, and all corresponding disturbances $d$ and $\tilde d$  and initial states $x(0)$ and $\tilde x(0)$, we have  $\tilde y(k)\not \rightarrow y(k)$, as $k\rightarrow \infty$.
\end{itemize}
\end{defn}

This notion is weaker than identifiability since an attack $a$ can be $i$-identifiable even if there exists
another attack $\tilde a \neq a$, with $a_i=\tilde a_i$, such that $y=\tilde y$.
Hence, $i$-identifiability concerns only the sensitivity of the output $y$ with respect to
 the $i$-th element in $a$. Identifiability is therefore the same as $i$-identifiability for all $i$.
 Interestingly, there is a tight connection between detectability, identifiability, and the
 previously introduced security index.

\begin{thm}
Suppose that the initial state $x(0)$ is unknown to the operator (and can take any value), and that the attacker can manipulate at most $q$
attack elements simultaneously ($\|a\|_0\leq q$).
\begin{itemize}
\item[(i)] There exists persistent undetectable attacks $a^i$ iff $q\geq \alpha_i$;
\item[(ii)] All persistent attacks are $i$-identifiable iff $q < \alpha_i/2$;
\item[(iii)] All persistent attacks are identifiable iff $q < \min_i \alpha_i/2$.
\end{itemize}
\label{thm:alpha}
\end{thm}
\begin{proof}
(i): Follows directly from the definition of $\alpha_i$, where we pick $x(0)=x_0$.
(ii): Consider first two attacks $a$ and $\tilde a$, both of cardinality $q<\alpha_i/2$, such that $a_i\neq \tilde a_i$. Let $y=G_d d +G_a a$ and $\tilde y = G_d \tilde d + G_a \tilde a$ and suppose that $y = \tilde y$, in contradiction to the theorem. This implies that $0 = G_d (d-\tilde d)+ G_a (a-\tilde a)$. Since $a_i\neq \tilde a_i$, the attack signal $a-\tilde a$ would constitute an undetectable attack against element $i$. Furthermore,
the cardinality of this signal is strictly smaller than $\alpha_i/2 + \alpha_i/2$, which is a contradiction to the optimality of security index $\alpha_i$. Hence, we must have $\tilde y \neq y$, and the attack $a$ is $i$-identifiable. Conversely, assume that
$q\geq \alpha_i/2$ and let us construct two attacks $a$ and $\tilde a$ that are not $i$-identifiable. Suppose first that $\alpha_i$ is even and that $q=\alpha_i/2$. There exists an undetectable attack $a^\star$, targeting element $i$, with support in an index set $I$, $|I|=\alpha_i$. Thus $0 = G_d d^\star + G_a a^\star$.
Let us split $I$ into two disjoint sets, $J$ and $K$ of equal size, $I=J\cup K$, $|J|=|K|=\alpha_i/2$. In a corresponding manner we can make the split $a^\star = a - \tilde a$, where $a$ and $\tilde a$ have support in $J$ and $K$, respectively. It is now clear that
$0 \neq  y = G_d d^\star + G_a a = G_a \tilde a = \tilde y$, and since $a_i\neq \tilde a_i$ this is an example of a non $i$-identifiable attack $a$. A similar argument can be applied when $\alpha_i$ is odd, concluding the proof.
(iii): Follows since identifiability is the same as $i$-identifiability for all $i$.
\end{proof}

In some cases it may be more realistic to assume that the operator actually knows the
initial state of the system~\eqref{eq:G}. We can then state the following corollary to the above theorem, which
applies in the asymptotic limit when $k\rightarrow \infty$.

\begin{cor}
Suppose that $A$ is Schur, that the initial state $x(0)$ is known to the operator, and that the attacker can manipulate at most $q$
attack elements simultaneously ($\|a\|_0\leq q$).
\begin{itemize}
\item[(i)] There exists persistent asymptotically undetectable attacks $a^i$ iff $q\geq \alpha_i$;
\item[(ii)] All persistent attacks are asymptotically $i$-identifiable iff $q < \alpha_i/2$;
\item[(iii)] All persistent attacks are asymptotically identifiable iff $q < \min_i \alpha_i/2$.
\end{itemize}
\label{cor:q}
\end{cor}
\begin{proof}
The only difference to the proof of Theorem~\ref{thm:alpha}
is that we need to add a transient term $y_\text{trans}(k) = CA^{k}(x(0)-x_0)$
 to all outputs, see~\cite{Teixeira_Allerton2012}. Here $x_0$ is an initial state rendering the relevant attack undetectable.
  Since $\rho(A)<1$ by assumption, this term decays to zero exponentially and the asymptotic results follow.
\end{proof}

We note that other papers have previously pointed out the connection
between detectability and identifiability of attacks, see, \eg \cite{pasqualetti+13}. The main contribution here is to introduce $i$-identifiability and show the relation to the security index $\alpha_i$.
As an example, assume that $\alpha_1 = 1$, $\alpha_2 = 3$, and that $q=1$. Then there will exist attacks against $a_1$ that are not visible in $y$, but all attacks against $a_2$ will not only be visible but also identifiable through $y$.
How to possibly conduct the identification is discussed next.

\subsection{Decoupling the Attacks from the Disturbances}
\label{sec:decouple}
To identify attacks $a$ in the output $y$, there are several useful results in the fault detection literature, see, \eg \cite{patton+99,ding08}. In particular, we will use a result on the existence of decoupling filters, which isolate the influence of the attack from that of the disturbance. A key result towards identification is the existence of certain left inverses.

\addtolength{\textheight}{-5.5cm}   % This command serves to balance the column lengths

\begin{defn}
\label{def:leftinv}
Consider the linear system $y=Gu$ with $m$ inputs, $p$ outputs, and with realization
\begin{equation*}
\begin{aligned}
x(k+1) & = Ax(k)+Bu(k) \\
y(k) & = Cx(k) + Du(k).
\end{aligned}
\end{equation*}
Then
$G$ has a \emph{left inverse} when $y(k)=0$, $k\geq 0$, implies that $u(k)=0$, $k \geq 0$, provided $x(0)=0$.
\end{defn}

The following condition for existence of a left inverse is well known, see, \eg \cite{moylan77,hou+98}.
\begin{lem}
\label{lem:leftinv}
A linear system $G\in\mathcal{R}_p^{p\times m}(z)$
has a \emph{left inverse} iff $\normrank G(z)=m$.
\end{lem}

From fault detection \cite{ding08}, it is known that if $G_d,G_a\in \mathcal{R}_{p}(z)$ and
\begin{equation}
\begin{aligned}
\normrank [G_d(z)] & = m', \\  \normrank [G_d(z)\,\, G_a(z)] & = m'+m'',
\end{aligned}
\label{eq:rankcond}
\end{equation}
then there exists a post-filter $R\in \mathcal{R}_{p}^{p\times p}(z)$ (of full normal rank) such that we can decouple the effects of the attacks and the disturbances in the following way:
\begin{equation}
\begin{bmatrix} r \\ y' \end{bmatrix} =R(G_d d+ G_a a) = \begin{bmatrix} 0 & \Delta \\ G_d' & G_a' \end{bmatrix} \begin{bmatrix} d \\ a \end{bmatrix},
\label{eq:decouple}
\end{equation}
where $\normrank [G_d'(z)]= \normrank [G_d'(z) \, G'_a(z)] = m'$ and $\normrank [\Delta(z)]=m''$. Note that if all attacks are undetectable in the sense of \eqref{eq:seriousrank}, then $m''=0$, and
$\Delta$ will be the empty matrix. On the other hand, if for some $i$, $\alpha_i>1$, then $m''>0$ and there is a non-trivial system $\Delta$.
The residual signal $r$ is only influenced by the attack $a$, and we can use it to detect and potentially identify $a$. Notice that for all attacks $a$ there exists a disturbance $d$ such that $0=y'=G_d'd+G_a'a$, so that $r=\Delta a$ is the only reliable source of information in regards to $a$. We have the following proposition on the relation between the measured output $y$ and
the filtered version $r$.

\begin{prop}
Let the initial state of the decoupling filter $R$ be chosen to $x_R(0)=0$. Suppose the initial state $x(0)$ is unknown to the operator (and can take any value), and that the attacker can manipulate at most $q$
attack elements simultaneously ($\|a\|_0\leq q$).
\begin{itemize}
\item[(i)] There exists persistent undetectable attacks $a^i$ in the signal $r$ iff $q\geq \alpha_i$;
\item[(ii)] All persistent attacks are $i$-identifiable in the signal $r$ iff $q < \alpha_i/2$;
\item[(iii)] All persistent attacks are identifiable in the signal $r$ iff $q < \min_i \alpha_i/2$.
\end{itemize}
\label{prop:R}
\end{prop}
\begin{proof}
Recalling that $R$ has full normal rank, we can use Lemma~\ref{lem:leftinv} and Definition~\ref{def:leftinv} to conclude that
 $Ry=0$ is equivalent to $y=0$. Since there is always a $d$ such that $y'=0$ in \eqref{eq:decouple},
the undetectability and identifiability properties of $a$ in $y=G_d d+G_a a$ must
carry over to the relation $r=\Delta a$, to which we can apply Theorem~\ref{thm:alpha}.
\end{proof}

If we suppose that $q<\min_i \alpha_i/2$, all persistent attacks are identifiable.
 A procedure to identify $a$ could include the following steps (we leave the details for future work):
 First apply the post-filter $R$ to $y$ to obtain the relation $r = \Delta a$. The initial state $x(0)$ is unknown, and could cause a non-zero transient in $r$ even in the absence of an attack $a$. However, the dynamics of the transients are known, and can be filtered out from $r$ to obtain a new transient-free residual $r'$. The signal $r'$ is identically zero if $r$ can be completely explained by a transient $y_\text{trans}(k) = CA^{k}x(0)$.
Undetectable attacks could also be also "hiding" in the transient, and by forming $r'$ the visible effects of such possible attacks also disappear. However, since we know that $q<\min_i \alpha_i$, there are no such persistent attacks affecting $a$, and so to identify $a$ we can equally well use the relation $r' = \Delta a$, where the initial state of $\Delta$ is zero, $x_\Delta(0)=0$.
To find $a$, we can form the systems $\Delta_I:=[\Delta_i]_{i\in I}$ out of the columns $\Delta_i$ of $\Delta$, for all subsets
$|I|\leq q$, $I\subseteq \{1,\ldots,m\}$. Since all attacks are identifiable, these $\Delta_I$ are left invertible, and give each rise to an attack estimation $\hat a_I$. From identifiability of $a$ it follows that any estimate $\hat a_I$
satisfying $r'=\Delta \hat a_I$ is actually equal to the real persistent attack $a$, which concludes the procedure.

Note that the real bottleneck here is the number of systems $\Delta_I$ that need to be formed and inverted. The problem is in fact essentially the same as in compressed sensing, see, \eg \cite{donoho+03}. Finally, we remark that the procedure can be modified to handle attacks that are only $i$-identifiable, but the estimates $\hat a_I$ will then only necessarily correctly identify element $a_i$.

\section{Conclusion}
\label{sec:conclusion}
In this paper, we have studied detectability and identifiability of attacks on dynamical systems that are also subjected
to disturbances. For this purpose, we generalized the concept of security index, which was previously introduced for
static systems in~\cite{sandberg+10}. In particular, the index exactly quantifies the resources necessary for targeted attacks to be undetectable and unidentifiable in the presence of disturbances.
Such information is relevant for both risk assessment and for the design of anomaly detectors.
We also discussed how these concepts relate to recent other work on attack detection and identification. Finally, we showed how techniques from the fault detection literature can be exploited to identify attacks under certain sparsity constraints.

\bibliographystyle{IEEEtran}
%\bibliography{SoSCYPS-refs}

\end{document}